\documentclass{svproc}
\usepackage{url}

\usepackage{balance}
\usepackage{hyperref} 
\usepackage[T1]{fontenc}
\usepackage{lmodern}
\usepackage{multicol}
\usepackage{footmisc}
\usepackage{amssymb}
\usepackage{amsmath} 
\usepackage[table]{xcolor}
\usepackage{cleveref}
\usepackage[utf8]{inputenc}
\usepackage{xcolor}
\usepackage{graphicx}
\usepackage{amsfonts}
\usepackage{nicefrac}
\usepackage{geometry}
\usepackage{nicefrac}
 
\begin{document}
\mainmatter

\title{A Social Choice Analysis of Optimism's Retroactive Project Funding}
\titlerunning{Optimism's Retroactive Project Funding} 

\author{
    Eyal Briman \inst{1} \and
    Nimrod Talmon \inst{1} \and
    Angela Kreitenweis \inst{2} \and
    Muhammad Idrees \inst{3}
}
\authorrunning{Briman et al.}

\institute{
    Ben-Gurion University, GovXS, Israel\\
    \email{\{briman@post.bgu.ac.il, nimrodtalmon77@gmail.com\}} \and
    GovXS, Germany\\
    \email{angela@tokenengineering.net} \and
GovXS, Pakistan\\
    \email{idrees535@gmail.com}
}

\maketitle

\begin{abstract}
The Optimism Retroactive Project Funding (RetroPGF) is a key initiative within the blockchain ecosystem that retroactively rewards projects deemed valuable to the Ethereum and Optimism communities. Managed by the Optimism Collective, a decentralized autonomous organization (DAO), RetroPGF represents a large-scale experiment in decentralized governance. Funding rewards are distributed in OP tokens, the native digital currency of the ecosystem. As of this writing, four funding rounds have been completed, collectively allocating over \$100M, with an additional \$1.3B reserved for future rounds. However, we identify significant shortcomings in the current allocation system, underscoring the need for improved governance mechanisms given the scale of funds involved. 

Leveraging computational social choice techniques and insights from multiagent systems, we propose improvements to the voting process by recommending the adoption of a utilitarian moving phantoms mechanism~\cite{freeman2019truthful}. This mechanism was originally introduced by Freeman et al. (2019), is designed to enhance social welfare (using the $\ell_1$ norm) while satisfying strategyproofness—two key properties aligned with the application's governance requirements. Our analysis provides a formal framework for designing improved funding mechanisms for DAOs, contributing to the broader discourse on decentralized governance and public goods allocation.

\end{abstract}

\section{Introduction}

The Optimism Retroactive Public Goods Funding (RetroPGF) initiative, managed by the Optimism Collective—a decentralized autonomous organization (DAO)—is a prominent funding mechanism within the blockchain ecosystem.\footnote{\url{https://retrofunding.optimism.io/}} 
It operates within a broader economic framework designed to incentivize innovation and retroactively reward projects that have demonstrably contributed to the Ethereum and Optimism ecosystems. The initiative has undergone four funding rounds~\cite{retro1,retro2,retro3,retro4}, with this study conducted in the context of the ongoing fifth round~\cite{retro5}.
Unlike forward-looking grant systems, RetroPGF distributes funds \textit{ex post}, rewarding projects based on their verified impact rather than projected outcomes. This approach, increasingly adopted in decentralized ecosystems, addresses challenges in evaluating contributions in real time. 

To date, RetroPGF has distributed over \$100M, with \$1.3B earmarked for future rounds.\footnote{\url{https://optimism.mirror.xyz/nz5II2tucf3k8tJ76O6HWwvidLB6TLQXszmMnlnhxWU}} Given the scale of resources involved, the governance structure of the funding process plays a crucial role in ensuring fair and effective allocation. However, we identify several limitations in the election system currently used, motivating the need for improved voting mechanisms.

\paragraph{Research Approach.}
In this work, we employ computational social choice to analyze and enhance the RetroPGF voting system. Specifically, based on formalized governance requirements derived from public discussions and direct engagement with the Optimism Collective, we propose a theoretically grounded and empirically tested improvement. Our analysis conceptualizes RetroPGF as a majority-based, ground-truth revealing, token-based allocation process, incorporating domain-specific constraints from blockchain governance. 

To address observed inefficiencies, we propose adopting the utilitarian moving phantoms mechanism~\cite{freeman2019truthful}, which balances strategyproofness with social welfare maximization under the $\ell_1$ norm. Our study consists of two primary components: 
\begin{itemize}
    \item A theoretical analysis of both the voting rules used in past RetroPGF rounds and the proposed alternative voting rules, along with their formal properties.
    \item A simulation-based evaluation of previous rules and an alternative rule to assess their practical performance under realistic voter behavior models.
\end{itemize}

\subsection{Our Contributions}
This paper advances the study of decentralized funding mechanisms, with a focus on the RetroPGF process in Optimism. Our key contributions are as follows:

\begin{itemize}
    \item \textbf{Formalization of Governance Requirements:} We provide the first formal mapping of Optimism's ideological and practical desiderata into social choice criteria, encompassing both qualitative and quantitative dimensions.
    \item \textbf{Comparative Analysis of Voting Mechanisms:} We conduct a rigorous theoretical and empirical evaluation of existing and proposed voting rules within the RetroPGF framework.
    \item \textbf{Proposal for an Improved Voting Rule:} We introduce and analyze the moving phantoms mechanism, demonstrating its advantages in strategyproofness, fairness, and social welfare maximization.
\end{itemize}

\subsection{Paper Structure}

The paper is structured as follows:
We begin with a discussion on the informal requirements by Optimism, given by their forums and from personal communication (Section~\ref{section
specific}).
We then discuss related work (Section~\ref{section:related work}).
We continue to describe a formal model of RetroPGF (Section~\ref{section:formal_model}).
Then, we discuss voting rules used for RetroPGF (Section~\ref{section:voting_rules}).
Concrete evaluation metrics are discussed in Section~\ref{section:metrics}.
We then report on our theoretical analysis (Section~\ref{theoretical}) and on our simulation-based analysis (Section~\ref{section:experimental} and Section~\ref{section:results}).
We conclude with a discussion (Section~\ref{section:outlook}).

\section{Application Requirements}\label{section
specific}

Based on discussions from Optimism’s public forums and direct communication with key stakeholders, we identified the specific application requirements for the RetroPGF process. The primary goal is to compare and evaluate the behavior of different voting rules in the context of these requirements. Our analysis is framed around key dimensions outlined in Optimism's governance documentation\footnote{\url{https://gov.optimism.io/t/the-future-of-optimism-governance/6471}} and tailored to the unique characteristics of this decentralized funding mechanism.

We note that these requirements reflect a blend of common-value elements---where all voters aim to reward impactful projects---and private-value elements, stemming from heterogeneity in voters’ domain expertise, interpretation of “impact,” and strategic considerations. As a benchmark, we highlight that if all badgeholders had fully aligned beliefs and no individual bias, a symmetric preference profile would emerge, resulting in unanimous allocations. We return to this benchmark in later discussion.

\begin{itemize}

\item \textbf{Badgeholder (voter) Responsibility}: A key feature of the RetroPGF process is the delegation of fund allocation to a small group of certified badgeholders (voters). These individuals possess a reputation within the ecosystem and are selected by the collective to distribute tokens to projects based on assessed needs. Badgeholders are expected to represent the broader interests of the community while maintaining impartiality through a strict code of conduct, including conflict of interest protocols. Each funding round is targeted at a specific context. For example, in Round 6, badgeholders evaluated governance-oriented projects such as Delegation Analytics and Agora, while in Round 4 the focus included infrastructure and public goods more broadly.

\item \textbf{Majoritarian Decision-Making}: The RetroPGF process emphasizes a majority-based voting approach. The goal is to capture the ``ground truth'' of which projects have contributed the most value to the ecosystem. Badgeholders play a pivotal role in aligning their token distribution decisions with collective goals. For instance, projects that receive broad support across badgeholders are presumed to reflect shared community values.

\item \textbf{Iterative Decentralization}: A core principle of the system is its ability to adapt and improve through iteration. As the RetroPGF process evolves, decentralization becomes more prominent. This iterative learning mechanism resembles adaptive systems in social choice theory, where feedback loops (such as post-round reports, retrospective debates, and proposal design changes) help refine collective decisions over time.

\item \textbf{Equity and Balance}: Optimism seeks to balance governance power, ensuring that financial influence does not overshadow community input. In Retroactive Project Funding, one person---or more precisely, a pseudonymous wallet holding a badgeholder token---equals one vote, with equal voting power for all voters. This design aligns with social choice concepts of fairness, aiming to avoid plutocratic control and to promote outcomes that serve the ecosystem as a whole. For example, the cap on the number of badgeholders per funding round reinforces equal voice in decision-making.

\item \textbf{Impact = Profit}: The RetroPGF mechanism rewards projects based on their demonstrated contributions to the community. The guiding principle is that “impact should be rewarded,” and therefore projects creating ecosystem value should be financially supported. This principle parallels utilitarian and proportional notions in social choice theory. For instance, in Round 4, large allocations were awarded to projects like L2Beat and Gitcoin, reflecting their visible impact on infrastructure and community support.

\end{itemize}

These application requirements form the basis for our analysis of voting rules. In particular, we focus on how well various mechanisms perform with respect to efficiency, strategic resistance, and fairness under realistic voter models. While a symmetric preference scenario may lead to agreement on allocations, we argue that the presence of individual biases, reputational considerations, and information asymmetries necessitates a more nuanced modeling of badgeholder behavior.

\section{Related Work}\label{section:related work}

Research at the intersection of public good funding, computational social choice, and blockchain governance has developed significantly over the past decade. Early models such as participatory budgeting~\cite{aziz2021participatory} and quadratic funding~\cite{pasquini2020quadratic} have been adapted to decentralized contexts, raising new challenges in designing fair, efficient, and strategy-resistant allocation mechanisms. More recently, retroactive funding systems have emerged as a novel paradigm, where reward is based on verified impact rather than anticipated outcomes. This shift demands mechanisms that can robustly aggregate diverse preferences and uncover ground truth contributions---a setting where social choice theory, particularly in its algorithmic and strategic dimensions, becomes highly relevan~\cite{cohen1986epistemic}.

\paragraph{Public Good Funding}
Public good funding distributes resources to benefit communities, with decentralized models such as quadratic funding~\cite{buterin2019flexible,georgescu2024fixed,dimitri2022quadratic} amplifying smaller contributions to ensure fairness. Participatory budgeting~\cite{aziz2021participatory,fairstein2023participatory} allows voters to decide on resource allocation prior to implementation. RetroPGF, as used in Optimism, instead rewards projects based on their \emph{proven} impact, representing a shift from ex-ante to ex-post evaluation.

\paragraph{RetroPGF and Portioning}
Optimism’s RetroPGF evaluates completed projects to reward those with the most community value~\cite{zichichi2019likestarter}. This approach parallels the setting of \emph{portioning}~\cite{airiau2023portioning}, where resources are divided after evaluating outcomes. In both settings, the role of the voter is evaluative rather than predictive, and the aggregation mechanism must navigate heterogeneous signals of impact.

\paragraph{Optimism's RetroPGF}
Operating within a DAO framework, Optimism’s RetroPGF integrates blockchain transparency and social choice theory~\cite{talmon2023social}, leveraging majoritarian decision-making to align token distribution with perceived community value. The design space intersects with work on blockchain governance~\cite{jones2019blockchain} and computational social choice in digital settings~\cite{grossi2022social}, raising questions about fairness, robustness, and decentralization in large-scale decision-making.

\paragraph{Uncovering the Ground Truth}
A central challenge in RetroPGF is to accurately aggregate badgeholder votes to reveal the true impact of projects. Prior work in elicitation and aggregation under uncertainty, such as Bayesian Truth Serum~\cite{prelec2004bayesian,prelec2017solution} and noisy preference models~\cite{caragiannis2016noisy}, offer tools for designing mechanisms that align collective outcomes with objective contribution measures. These methods are particularly relevant for blockchain-based voting systems, where preference aggregation must also resist manipulation while remaining transparent.

\section{Formal Model}\label{section:formal_model}

We formalize the token-based Retroactive Project Funding process as a voting-based budget allocation mechanism. Let \( N = \{1, \dots, n\} \) be the set of voters and \( P = \{1, \dots, m\} \) the set of projects. The total available budget is \( B \), which we normalize to \( 1 \), ensuring that all allocations are expressed as fractions of the total budget. Each voter is allocated \( c \) tokens, with the constraint that the total number of tokens does not exceed the budget, i.e., \( c \cdot n \leq B \).

A feasible allocation is a vector \( \mathbf{a} = (a_1, \dots, a_m) \) where:
\[
a_p \geq 0, \quad \forall p \in P, \quad \text{and} \quad \sum_{p \in P} a_p \leq B.
\]

Each voter \( i \) submits a cumulative ballot \( X_i = (x_{i,1}, \dots, x_{i,m}) \), where \( x_{i,p} \geq 0 \) represents the fraction of their tokens allocated to project \( p \), subject to:
\[
\sum_{p \in P} x_{i,p} \leq c.
\]

Since voters benefit from fully utilizing their allocated tokens, we assume:
\[
\sum_{p \in P} x_{i,p} = c, \quad \forall i \in N.
\]

The final allocation \( \mathbf{a} \) is determined by a voting rule, formally defined as a function:
\[
f: \left( X_1, \dots, X_n \right) \to \mathbf{a} = (a_1, \dots, a_m),
\]
which satisfies the budget constraint:
\[
\sum_{p \in P} a_p \leq B.
\]

\paragraph{Preference Structure.} 
Although the RetroPGF process aims to reward projects according to their objective impact, we assume that badgeholders may exhibit heterogeneous preferences, modeled via different cumulative ballots. This heterogeneity reflects several practical considerations:
\begin{itemize}
    \item \emph{Epistemic diversity:} Voters possess varying knowledge about the ecosystem and specific projects, leading to different interpretations of what constitutes impact.
    \item \emph{Domain affinity:} Some voters may prioritize projects aligned with their expertise or values (e.g., governance vs. infrastructure).
    \item \emph{Reputational strategy:} Badgeholders might consider not only what is best for the ecosystem, but also how their votes will be perceived within the community.
\end{itemize}

As a theoretical benchmark, we consider the \emph{symmetric preference case}, in which all badgeholders submit the same ballot, i.e., \( X_1 = \dots = X_n \). This case serves as a useful baseline for evaluating how different rules behave under full alignment. However, since empirical data from past rounds (e.g., Round 4) shows clear variation in allocations, we argue that modeling badgeholder preferences as heterogeneous provides a more realistic and informative foundation for mechanism evaluation.

\section{Voting Rules}\label{section:voting_rules}

This section reviews the voting rules used in Optimism's RetroPGF rounds and introduces improved mechanisms.

\paragraph{Quadratic Voting (Round 1).} 
Quadratic Voting (QV)~\cite{lalley2018quadratic} allows voters to allocate \( x_{i,p} \) tokens to project \( p \), with the effective vote weight given by:
\[
a_p = \frac{\sum_{i \in N} \sqrt{x_{i,p}}}{\sum_{i \in N} \sum_{p' \in P} \sqrt{x_{i,p'}}}.
\]

\paragraph{Mean Rule (Round 2).} 
Allocations are proportional to the total tokens received:
\[
a_p = \frac{\sum_{i \in N} x_{i,p}}{\sum_{i \in N} \sum_{p' \in P} x_{i,p'}}.
\]

\paragraph{Quorum Median Rule (Round 3).} 
Allocations are based on the median vote with quorum constraints \( q_1 \) (minimum tokens) and \( q_2 \) (minimum voters):
\[
b_p = \frac{\text{median} \{ x_{i,p} \mid x_{i,p} > 0 \}}{\sum_{p' \in P} \text{median} \{ x_{i,p'} \mid x_{i,p'} > 0 \}}.
\]
If \( b_p \geq q_1 \) and at least \( q_2 \) voters contribute, then \( d_p = b_p \); otherwise, \( d_p = 0 \). The final allocation is:
\[
a_p = \frac{d_p}{\sum_{p'\in P} d_{p'}}.
\]

\paragraph{Capped Median Rule (Round 4).} 
A variant of the Quorum Median Rule, with an upper bound \( K_1 \) and redistribution of excess funds:
\[
c_p = \frac{\text{median} \{ x_{i,p} \}}{\sum_{p' \in P} \text{median} \{ x_{i,p'} \}} \cdot B, \quad
d_p = \min(c_p, K_1) + \frac{\sum_{j} \max(0, c_j - K_1) \cdot c_p}{\sum_{j} c_j}.
\]
Projects below \( K_2 \) are eliminated, and their allocations redistributed:
\[
b_p = 
\begin{cases}
0, & d_p < K_2, \\
d_p + \frac{\sum_{p'} d_{p'} < K_2}{\sum_{p'} d_{p'} \geq K_2} \cdot d_p, & \text{otherwise}.
\end{cases}
\]
Final normalization ensures:
\[
a_p = \frac{b_p}{\sum_{p'} b_{p'}}.
\]

\paragraph{Midpoint Rule.} 
The allocation minimizes the \( \ell_1 \) distance from all voter allocations~\cite{nehring2008allocating}:
\[
\mathbf{x} = X_{i^*}, \quad i^* = \arg\min_{i \in N} \sum_{j=1}^{n} \| X_i - X_j \|_1.
\]
\paragraph{Moving Phantoms.} 
The Moving Phantoms mechanism~\cite{freeman2019truthful} addresses budget inconsistencies in median-based voting by introducing phantom voters whose influence dynamically adjusts allocations while preserving strategyproofness. The final allocation is determined using a median-based adjustment:
\[
\mathcal{A}^F(a_p) = \text{med}(f_0(t^*), \dots, f_n(t^*), x_{1,p}, \dots, x_{n,p}),
\]
where \( t^* \) is chosen such that:
\[
f_0(t^*)+\sum_{i=1}^n  f_i(t^*)+\sum_{i=1}^n\sum_{p=1}^m x_{i,p}=1.
\]
This constraint ensures that the total allocation remains within budget. The optimal value of \( t^* \) can be computed efficiently using binary search.

Two variants of the Moving Phantoms mechanism are considered:

\paragraph{Independent Markets Algorithm.} 
In this variant, phantom influence follows a linear distribution:
\[
f_k(t) = \min\{ t(n - k), 1 \}.
\]

\paragraph{Majoritarian Phantoms Algorithm.} 
This variant prioritizes majority preferences by defining:
\[
f_k(t) = 
\begin{cases} 
0, & 0 \leq t \leq \frac{k}{n+1}, \\
t(n+1) - k, & \frac{k}{n+1} < t \leq \frac{k+1}{n+1}, \\
1, & \frac{k+1}{n+1} \leq t \leq 1.
\end{cases}
\].

\section{Properties and Metrics}\label{section:metrics}

To evaluate voting rules in the Retroactive Project Funding framework, we consider key theoretical properties and performance metrics. These criteria incorporate classical social choice principles~\cite{brandt2016handbook} and RetroPGF-specific requirements. To define these metrics and properties, we first establish the agent utility function based on the $\ell_1$ distance, as prior work commonly assumes that agents assess budget distributions by their 
$\ell_1$ distance from their ideal allocation~\cite{freeman2019truthful,de2024truthful}.

\paragraph{Resistance to Manipulation.}  
A voting rule should be robust against strategic behavior, including bribery and control. The \textit{cost of bribery} is the minimum expenditure required to increase project \( p \)'s allocation by \( X \) tokens:
\[
b = \sum_{i=1}^{n} \sum_{q \in P} |x_{i,q} - x'_{i,q}|,
\]
where \(V = \{X_1, \dots, X_n\}\) and \(V' = \{X_1', \dots, X_n'\}\) are the original and modified vote profiles, and \( \mathbf{x} \) is the outcome. The \textit{cost of control} measures the minimal number of voters that must be added or removed to increase \( p \)'s allocation by \( r \). A rule is \textit{robust} if the expected deviation in outcomes due to small vote perturbations remains bounded:
\[
\mathbb{E}[d(\mathbf{x}, \mathbf{x}')],
\]
where \( d \) is a distance metric such as \( \ell_1 \) or \( \ell_2 \). The \textit{Voter Extractable Value (VEV)} quantifies the maximum allocation shift a single voter can induce:
\[
\text{VEV} = \max_{i \in [n], k \in [m]} d(\mathbf{x}, \mathbf{x}^{(i,k)}),
\]
where \( \mathbf{x}^{(i,k)} \) is the outcome after voter \( i \) reallocates \( r\% \) of their vote to project \( k \).

\paragraph{Incentive Compatibility.}  
A voting rule is \textit{strategyproof} if truthful reporting is a weakly dominant strategy for all voters. Formally, for every \( i \in N \), let \( u_i(X_i, X_{-i}) \) be their utility when reporting truthfully and \( u_i(X_i', X_{-i}) \) their utility when submitting a strategic misreport \( X_i' \). The rule is strategyproof if:
\[
u_i(X_i, X_{-i}) \geq u_i(X_i', X_{-i}) \quad \forall i \in N, \quad \forall X_i' \neq X_i, \quad \forall X_{-i}.
\]

\paragraph{Outcome Quality.}  
A voting rule satisfies \textit{Pareto efficiency} if there exists no alternative allocation \( \mathbf{x}' \) that strictly improves the utility of at least one voter without making any other voter worse off. Formally, an allocation \( \mathbf{x} = (a_1, \dots, a_m) \) is Pareto efficient if:
\[
\forall \mathbf{x}' \in \mathbb{R}^m_{\geq 0}, \quad \exists i \in N \text{ such that } U_i(\mathbf{x}') > U_i(\mathbf{x}) \Rightarrow \exists j \in N \text{ such that } U_j(\mathbf{x}') < U_j(\mathbf{x}).
\]
This ensures that no allocation \( \mathbf{x}' \) is strictly better for all voters simultaneously.

A voting rule satisfies \textit{monotonicity} if increasing support for a project cannot decrease its allocation. That is, for any two voter profiles \( V \) and \( V' \), where \( V' \) is obtained by increasing the support for project \( p \) for all voters:
\[
f(V', p) \geq f(V, p) \quad \text{for } V' \text{ where } x'_{i,p} \geq x_{i,p} \text{ for all } i.
\]

The rule satisfies \textit{reinforcement} if combining two disjoint voter groups that yield the same outcome separately does not change the result:
\[
f(V_1 \cup V_2) = f(V_1) = f(V_2).
\]

\paragraph{Fairness and Representation.}  
A rule should balance majority rule with minority protection. 

\textit{Utilitarian social welfare} measures how well the outcome reflects voter preferences:
\[
W = \frac{1}{n} \sum_{i=1}^{n} \|\mathbf{x} - X_i\|_1.
\]

\textit{Proportionality} ensures that any subset \( S \subseteq N \) with \( |S| \geq n/k \) that exclusively supports a single project \( p \) guarantees \( p \) at least \( B/k \) of the total budget:
\[
\forall S \subseteq N, \quad |S| \geq \frac{n}{k}, \quad \forall i \in S, \quad x_{i,p} = c, \quad x_{i,p'} = 0 \text{ for } p' \neq p \Rightarrow a_p \geq \frac{B}{k}.
\]

Allocation inequality is quantified using the \textit{Gini index}, which measures the disparity in allocated resources:
\[
G = \frac{\sum_{i=1}^m \sum_{j=1}^m |a_i - a_j|}{2m \sum_{i=1}^m a_i}.
\]

\paragraph{Participation and Ground-Truth Alignment.}  
A rule satisfies \textit{participation} if no voter is worse off by voting, meaning that submitting a ballot cannot result in a lower utility than abstaining. Formally, for every voter \( i \), let \( \mathbf{x} = (a_1, \dots, a_m) \) denote the allocation when \( i \) does not participate, and let \( \mathbf{x}' = (a_1', \dots, a_m') \) be the allocation when \( i \) submits a ballot \( X_i \). The rule satisfies participation if:
\[
u_i(\mathbf{x}') \geq u_i(\mathbf{x}) \quad \forall i \in N.
\]
This ensures that participation is always beneficial or neutral.

\textit{Alignment with ground truth} measures the deviation from an objective allocation \( \mathbf{x^*} = (a_1^*, \dots, a_m^*) \), which represents the theoretically correct funding distribution assuming perfect voter expertise and full knowledge of project impact. The alignment metric is given by:
\[
d_{\text{GT}} = \sum_{i=1}^{m} |a_i^* - a_i|.
\]
This ensures that the allocation reflects the true contribution of funded projects~\cite{caragiannis2016noisy,cohensius2017proxy}.

\section{Theoretical Results}\label{theoretical}

We provide Table~\ref{table:voting_rules_combined}  that summarizes the theoretical results. For space considerations proofs are in the appendix. Note that we have added the Normalized median rule that is a generalization of R3 and R4 voting rules (since it does not include using of capping or quorum).

\begin{table*}[t]
\centering
\renewcommand{\arraystretch}{1.3} % Increase row height for readability
\scalebox{0.72}{
\begin{tabular}{|l|c|c|c|c|c|c|c|}
\hline
\textbf{Voting Rule} & \textbf{Reinforcement} & \textbf{Pareto- Efficiency} & \textbf{Monotonicity} & \textbf{Participation} & \textbf{Proportionality} & \textbf{Max Social Welfare} & \textbf{Strategyproofness} \\ \hline
R1 Quadratic  
    & \cellcolor{green!30} $\checkmark$           
    & \cellcolor{red!30} $\times$               
    & \cellcolor{green!30} $\checkmark$              
    & \cellcolor{red!30} $\times$           
    & \cellcolor{red!30} $\times$  
    & \cellcolor{red!30} $\times$  
    & \cellcolor{red!30} $\times$        
    \\ \hline
R2 Mean Rule      
    & \cellcolor{green!30} $\checkmark$          
    & \cellcolor{green!30} $\checkmark$              
    & \cellcolor{green!30} $\checkmark$           
    & \cellcolor{green!30} $\checkmark$           
    & \cellcolor{green!30} $\checkmark$    
    & \cellcolor{red!30} $\times$  
    & \cellcolor{red!30} $\times$        
    \\ \hline
R3 Quorum Median  
    & \cellcolor{red!30} $\times$           
    & \cellcolor{red!30} $\times$               
    & \cellcolor{red!30} $\times$            
    & \cellcolor{red!30} $\times$          
    & \cellcolor{red!30} $\times$        
    & \cellcolor{red!30} $\times$  
    & \cellcolor{red!30} $\times$        
    \\ \hline
R4 Capped Median Rule  
    & \cellcolor{red!30} $\times$           
    & \cellcolor{red!30} $\times$               
    & \cellcolor{green!30} $\checkmark$           
    & \cellcolor{red!30} $\times$          
    & \cellcolor{red!30} $\times$    
    & \cellcolor{red!30} $\times$  
    & \cellcolor{red!30} $\times$        
    \\ \hline
Normalized Median Rule          
    & \cellcolor{green!30} $\checkmark$           
    & \cellcolor{green!30} $\checkmark$               
    & \cellcolor{green!30} $\checkmark$
    & \cellcolor{green!30} $\checkmark$          
    & \cellcolor{red!30} $\times$   
    & \cellcolor{red!30} $\times$  
    & \cellcolor{red!30} $\times$        
    \\ \hline
Midpoint Rule        
    & \cellcolor{green!30} $\checkmark$                     
    & \cellcolor{green!30} $\checkmark$                          
    & \cellcolor{green!30} $\checkmark$                     
    & \cellcolor{green!30} $\checkmark$                     
    & \cellcolor{red!30} $\times$        
    & \cellcolor{red!30} $\times$  
    & \cellcolor{red!30} $\times$        
    \\ \hline
Independent Market~\cite{freeman2019truthful}   
    & \cellcolor{green!30} $\checkmark$          
    & \cellcolor{red!30} $\times$  for ($m \geq n^2$)
    & \cellcolor{green!30} $\checkmark$          
    & \cellcolor{green!30} $\checkmark$        
    & \cellcolor{green!30} $\checkmark$    
    & \cellcolor{red!30} $\times$  
    & \cellcolor{green!30} $\checkmark$ (for $m>2$)     
    \\ \hline
Majoritarian Phantom~\cite{freeman2019truthful}   
    & \cellcolor{green!30} $\checkmark$           
    & \cellcolor{green!30} $\checkmark$               
    & \cellcolor{green!30} $\checkmark$              
    & \cellcolor{green!30} $\checkmark$         
    & \cellcolor{red!30} $\times$    
    & \cellcolor{green!30} $\checkmark$   
    & \cellcolor{green!30} $\checkmark$ (for $m>2$)    
    \\ \hline
\end{tabular}
}
\caption{Properties of Voting Rules.}
\label{table:voting_rules_combined}
\end{table*}

\section{Experimental Design}\label{section:experimental}

Experiments were conducted using artificial voter data and a simplified version of Optimism's Round 4 dataset, the only fully available real-world Optimism voting dataset, in which 108 badge-holders collectively allocated 8 million tokens across 229 projects. The analysis compares the voting rules employed in Optimism Rounds 1-4 with the Majoritarian Phantoms mechanism, selected for its capacity to maximize $\ell_1$-based social welfare while ensuring strategyproofness—two fundamental properties aligned with the application's requirements.

\paragraph{Vote Generation.}
Cumulative ballot instances were generated using Mallows' model~\cite{boehmer2021putting,szufa2020drawing}, ensuring structured yet diverse voter preferences. The process is as follows:

\begin{enumerate}
    \item A base vote \( X_{\text{base}} \) is sampled from a Dirichlet distribution with parameters \( \alpha^m = 1^m \), ensuring it sums to 1:
    \[
    \sum_{p \in P} X_{\text{base},p} = 1.
    \]
    
    \item An independent vote \( X_{\text{independent}} \) is sampled from the same Dirichlet distribution (\( k = 1 \)).
    
    \item Each voter's ballot is computed as a weighted combination:
    \[
    X_i = 0.5 \cdot X_{\text{base}} + 0.5 \cdot X_{\text{independent}}.
    \]
\end{enumerate}

\paragraph{Experimental Setup.}  
To evaluate different voting rules, we conduct simulations across multiple scenarios, each designed to assess specific properties of the mechanisms. The experimental conditions are as follows:

\begin{itemize}
    \item Bribery, Control, Robustness, and Voter Extractable Value (VEV) Experiments: 
    These experiments use a setup with 40 voters distributing 8 million OP tokens across 145 projects.
    
    \item Social Welfare, Gini Index, and Alignment Experiments:  
    To analyze broader allocation trends, these experiments use an expanded setting with 145 voters, 600 projects, and 30 million OP tokens.  

    \item Trial Averaging: 
    Each experimental condition is repeated for 100 independent trials to ensure statistical reliability and minimize variance.
\end{itemize}

\paragraph{Resistance to Manipulation.}  
To examine the vulnerability of each voting rule to strategic behavior, we conduct two experiments:

\begin{enumerate}
    \item Control Experiment:  
    This experiment assesses how resistant a voting rule is to the addition or removal of voters. We measure the minimum number of strategically placed voters required to increase the funding of a specific target project by a predetermined percentage. The experiment systematically varies the target funding increase from 1\% to 30\% and runs 10 independent trials per setting.

    \item Bribery Experiment: 
    This experiment quantifies the cost of influencing the outcome by reallocating tokens. Assuming a unit cost per token reallocated, we measure the minimum expenditure required to shift funding in favor of a given project. Similar to the control experiment, we vary the targeted funding increase from 1\% to 30\% and conduct 10 trials per condition.
\end{enumerate}

\paragraph{Robustness and Voter Influence.}  
These experiments assess the stability of allocations under small perturbations and the extent of influence exerted by individual voters:

\begin{enumerate}
    \item Robustness Experiment: 
    We introduce controlled random variations to individual voter preferences and measure the impact on the final allocations. The robustness of a voting rule is quantified by computing the $\ell_1$ distance between the original allocation and the perturbed allocation. This process is repeated over 100 trials to evaluate stability across different scenarios.

    \item Voter Extractable Value (VEV) Experiment: 
    This experiment measures the maximum impact a single voter can exert on the allocation. A voter is allowed to concentrate between 90\% and 99\% of their total tokens on a single project, and we observe the resulting change in the project's final funding allocation. The goal is to quantify how susceptible each rule is to concentrated voting power.
\end{enumerate}

\paragraph{Alignment with Ground Truth.}  
To assess how well each voting rule reflects the true distribution of voter preferences, we compute the $\ell_1$ distance between the final allocation produced by a voting rule and a benchmark reference distribution. We set the benchmark to be the base vote used in the generation of all ballots.

\section{Experimental Results}\label{section:results}

We present Figure~\ref{fig:R4_Com}, which summarizes the key findings from our simulations on bribery costs, control, robustness, and VEV using Round 4 data (108 voters, 229 projects, 8M tokens).

\begin{figure*}[t]
    \centering
    \includegraphics[width=0.95\textwidth]{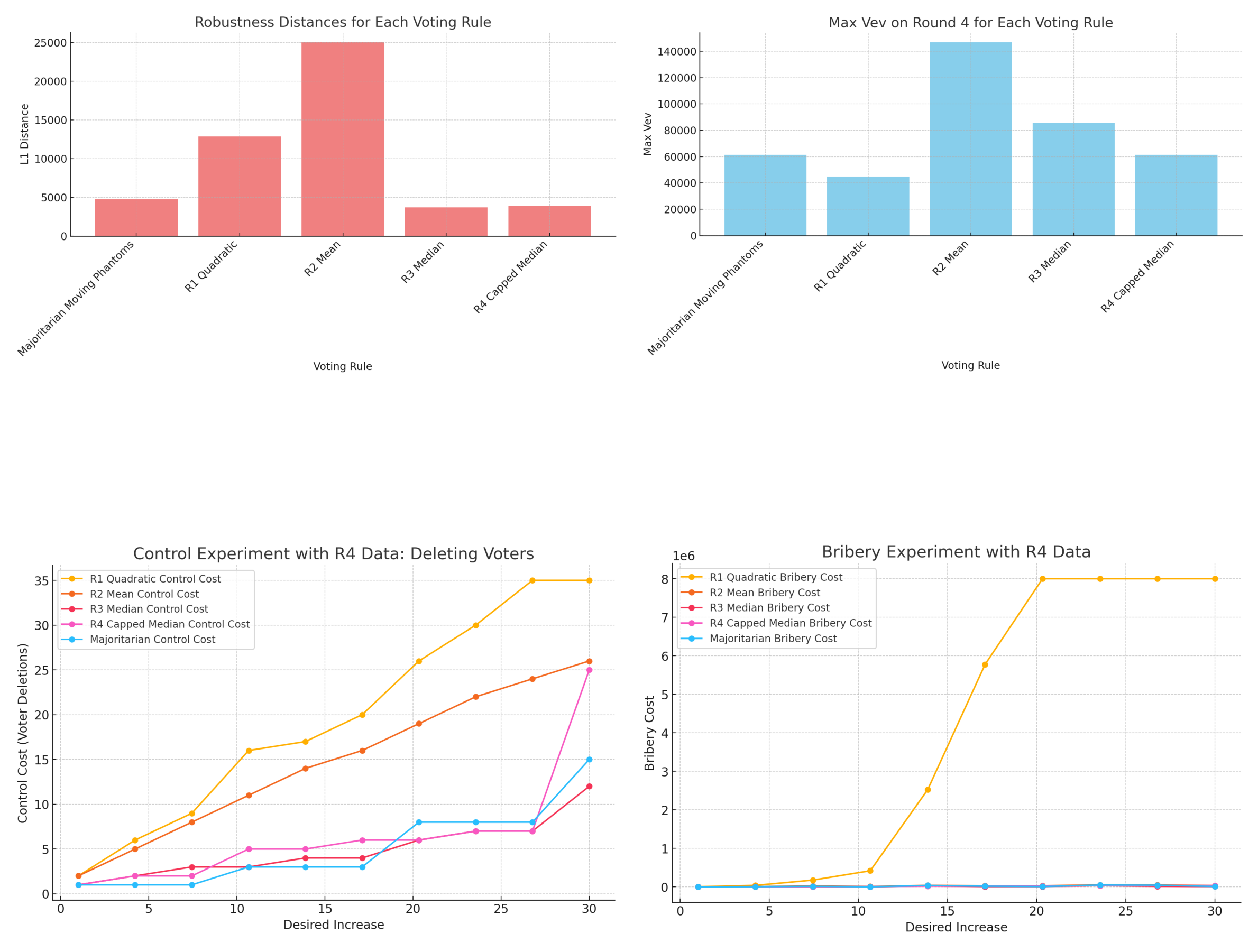}
    \caption{Bribery costs, cost of control by deleting voters, robustness, and Max VEV for each voting rule using Round 4 data.}
    \label{fig:R4_Com}
\end{figure*}

Majoritarian Phantom outperforms other rules across strategic robustness metrics, including higher bribery and control costs, stronger stability under noise, and reduced voter extractable value. However, as shown in the other figures in the appendix, the advantage is less pronounced for metrics like Gini index and alignment with ground truth. In those cases, differences between Majoritarian Phantom and the R3/R4 median-based rules are often small and not statistically significant, indicating that the practical gains, while consistent, may be incremental rather than large.

\section{Discussion and Outlook}\label{section:outlook}

Our theoretical and simulation results indicate that the Majoritarian Phantom rule is the most effective mechanism for RetroPGF among those tested, based on the application requirements set in Section~\ref{section specific}. It maximizes $\ell_1$ social welfare, satisfies key axioms such as Pareto and Participation, maintains a high Gini index (favoring impactful projects over equal distribution), and exhibits strong resistance to manipulation. It also aligns well with the ground truth—although it ranks third in this metric after Quadratic and Mean rules. However, as shown in our simulations, the empirical advantage of Majoritarian Phantom over median-based rules (R3 and R4) is often modest. In fairness and alignment metrics, the differences are typically small and not statistically significant. This suggests that while the rule offers strong theoretical benefits, its practical improvements may be incremental.

Our analysis focuses on the voting stage of the allocation process, assuming fixed distributions of voting tokens. We do not model pre-voting dynamics, such as trading or influence between badgeholders. In settings where more biased voters accumulate voting power through such exchanges, the choice of voting rule may interact with deeper strategic incentives. While Majoritarian Phantom remains strategy-resistant in the voting stage, its resilience under pre-voting manipulation is an open question warranting further study.

In addition, our current framework does not account for dynamic reputation costs. In practice, badgeholders face reputational consequences for misaligned or self-serving behavior, which could deter manipulation even without formal enforcement. Modeling such costs explicitly could improve the realism of strategic analyses and explain observed behavior in real-world rounds.

We formalized RetroPGF as a majority-based, token-weighted allocation process and analyzed its voting mechanisms through the lens of computational social choice. Our results contribute a structured framework for evaluating decentralized funding mechanisms, identifying tradeoffs between fairness, efficiency, and strategic robustness. We conclude with several directions for future work.

\paragraph{Comparison with Other Retroactive Project Funding Models.}  
Other ecosystems such as Filecoin’s RetroPGF\footnote{\url{https://filecoin.io/blog/posts/unveiling-fil-retropgf-1-retroactively-funding-filecoin-public-goods/}} use similar mechanisms. A comparative analysis could reveal institutional design insights across DAOs.

\paragraph{Metric-Based, Indirect Voting.}  
Optimism is experimenting with metric-based voting,\footnote{\url{https://gov.optimism.io/t/experimentation-impact-metric-based-voting/7727}} where badgeholders vote on evaluation criteria rather than projects. This may reduce cognitive load and limit manipulation. Evaluating this model is a promising research direction.

\paragraph{Mechanism Design Refinements.}  
Future refinements include reputation-weighted voting~\cite{leonardos2020weighted,de2020blockchain}, mitigation of pseudonymity risks~\cite{perez2018improving}, moving phantoms~\cite{de2024truthful}, and adaptations of VCG-like or Continuous Thiele rules~\cite{wagner2023strategy,wagner2024distribution}. Bayesian truth serum methods~\cite{prelec2004bayesian,prelec2017solution} may further help validate whether a rule meaningfully uncovers ground truth.

\bibliographystyle{plain}
\bibliography{bib}

\appendix
\section{Missing Proofs}
\subsection{Reinforcement}

\begin{theorem} The Mean, Normalized Median, and Midpoint Rules satisfy reinforcement. \end{theorem}

\begin{proof} 
For two profiles \( V_1 \) and \( V_2 \), let:
\[
a_p^1 = \frac{\sum_{i \in N_1} x_{i,p}}{\sum_{i \in N_1} \sum_{p' \in P} x_{i,p'}}, \quad 
a_p^2 = \frac{\sum_{i \in N_2} x_{i,p}}{\sum_{i \in N_2} \sum_{p' \in P} x_{i,p'}}.
\]
Since allocation summation is linear (Mean Rule) or preserves order statistics (Median Rule), the combined profile maintains relative rankings: 
\[
a_p^3 = \frac{\sum_{i \in N_1 \cup N_2} x_{i,p}}{\sum_{i \in N_1 \cup N_2} \sum_{p' \in P} x_{i,p'}}.
\]
For the Midpoint Rule, the voter minimizing \(\ell_1\) remains unchanged under profile union, ensuring reinforcement.  
\end{proof}

\begin{corollary} The Quadratic Rule satisfies reinforcement as it follows the Mean Rule after square-root transformation. \end{corollary}

\begin{example} The R3 Quorum Median and R4 Capped Median Rules violate reinforcement due to quorum effects.  
Consider two projects \( p_1, p_2 \) with quorum \( 0.2 \).  
- \( V_1 \): \( v_1 \) votes \( [0.9, 0.1] \rightarrow\)  Winner: \( [1,0] \).
- \( V_2 \): \( v_2 \) votes \( [0.9, 0.1] \rightarrow\) Winner: \( [1,0] \).
- \( V_1 \cup V_2 \): Quorum applies, shifting the result to \( [0.8,0.2] \).  
\end{example}

\subsection{Pareto Efficiency}

\begin{theorem}\label{theorem:mean_pareto}
The Mean Rule satisfies Pareto efficiency.
\end{theorem}

\begin{proof}
The Mean Rule distributes the budget proportionally among projects. Any reallocation that benefits one voter must reduce another project's allocation due to the budget constraint, making at least one voter worse off. Thus, no Pareto improvement is possible, proving Pareto efficiency.
\end{proof}

\begin{theorem}\label{theorem:median_pareto}
The Normalized Median Rule satisfies Pareto efficiency.
\end{theorem}

\begin{proof}
The median allocation minimizes absolute deviations from voter preferences. Increasing a project's allocation necessarily reduces another's, harming at least one voter. Thus, no Pareto improvement exists, proving Pareto efficiency.
\end{proof}

\begin{theorem}\label{theorem:midpoint_pareto}
The Midpoint Rule satisfies Pareto efficiency.
\end{theorem}

\begin{proof}
The selected allocation corresponds to a voter's input. Any alternative would increase some voters’ distance from their preferred allocation, making them worse off, ensuring Pareto efficiency.
\end{proof}

\begin{example}\label{example:quadratic_pareto}
The Quadratic Rule does not satisfy Pareto efficiency.

A voter allocating \([0.7, 0.3]\) receives an outcome based on square roots, deviating from their preference and allowing a Pareto improvement.
\end{example}

\begin{example}\label{example:quorum_pareto}
The R3 Quorum Median Rule does not satisfy Pareto efficiency.

A voter allocating \([0.7, 0.3]\) with a quorum of \( 0.4 \) results in \([1, 0]\), which is strictly worse for the voter.
\end{example}

\begin{corollary}\label{example:R4_pareto}
The R4 Capped Median Rule does not satisfy Pareto efficiency due to the quorum constraint \( K_2 \).
\end{corollary}

\subsection{Monotonicity}

\begin{theorem}\label{theorem:mean_monotonicity}
The Mean Rule satisfies monotonicity.
\end{theorem}

\begin{proof}
Increasing a voter's allocation to project \( p \) increases the numerator in:
\[
a_p = \frac{\sum_{i \in N} x_{i,p}}{\sum_{i \in N} \sum_{p' \in P} x_{i,p'}}
\]
while the denominator grows less significantly, ensuring \( a_p \) does not decrease.
\end{proof}

\begin{theorem}\label{theorem:median_monotonicity}
The Normalized Median Rule satisfies monotonicity.
\end{theorem}

\begin{proof}
Increasing \( x_{i,p} \) either shifts the median upward or leaves it unchanged. Thus, \( a_p \) does not decrease.
\end{proof}

\begin{theorem}\label{theorem:midpoint_monotonicity}
The Midpoint Rule satisfies monotonicity.
\end{theorem}

\begin{proof}
If a voter increases their allocation to \( p \), it reduces their \( \ell_1 \)-distance to the Midpoint outcome, preventing a worse allocation.
\end{proof}

\begin{corollary}\label{corollary:quadratic_monotonicity}
The Quadratic Rule satisfies monotonicity.
\end{corollary}

\begin{theorem}\label{theorem:quorum_monotonicity}
The R3 Quorum Median Rule does not satisfy monotonicity.
\end{theorem}

\begin{example}
Voters \( v_1, v_2, v_3 \) cast \([0.5,0.5]\), \([0,1]\), \([1,0]\), yielding \([0.5,0.5]\). If \( v_2 \) changes to \([0.1,0.4]\), the new allocation is \([0.4,0.6]\), violating monotonicity.
\end{example}

\begin{theorem}\label{theorem:R4_monotonicity}
The R4 Capped Median Rule satisfies monotonicity.
\end{theorem}

\begin{proof}
If a project's allocation is below the cap \( K_1 \), additional votes increase \( a_p \). If at or above \( K_1 \), it remains unchanged. If below quorum \( K_2 \), additional votes push \( a_p \) above quorum. If a project falls below quorum, fewer competing projects benefit from redistribution. Thus, monotonicity holds.
\end{proof}

\subsection{Participation}

\begin{theorem}\label{theorem:mean_participation}
The Mean Rule satisfies participation.
\end{theorem}

\begin{proof}
If voter \( i \) abstains, their influence on allocations is removed, generally resulting in a worse outcome for them, as their preferred projects may receive less funding. Thus, abstaining never increases utility.
\end{proof}

\begin{theorem}\label{theorem:median_participation}
The Normalized Median Rule satisfies participation.
\end{theorem}

\begin{proof}
Abstaining shifts the median in a way that can only worsen or maintain the voter’s preferred outcome. If \( x_{i,p} \) is below the median, abstaining does not improve it; if above, abstaining reduces the median allocation.
\end{proof}

\begin{theorem}\label{theorem:midpoint_participation}
The Midpoint Rule satisfies participation.
\end{theorem}

\begin{proof}
If the voter’s ballot is selected as the midpoint, participation directly benefits them. Otherwise, it follows from the Mean Rule proof, as their participation reduces total \( \ell_1 \)-distance.
\end{proof}

\begin{example}\label{example:quadratic_participation}
The Quadratic Rule does not satisfy participation.

A voter with \( [0.8, 0.2] \) in a setting where the initial allocation is also \( [0.8, 0.2] \) may, by participating, shift the outcome away due to the square root transformation.
\end{example}

\begin{example}\label{example:quorum_participation}
The R3 Quorum Median Rule does not satisfy participation.

A project \( p_2 \) is close to funding but does not meet quorum. If voter \( v \) adds a small vote \( [1-\epsilon, \epsilon] \), \( p_2 \) suddenly receives significant funding, potentially harming \( v \).
\end{example}

\begin{corollary}\label{example:R4_participation}
The R4 Capped Median Rule does not satisfy participation as it includes a Quorum \( K_2 \).
\end{corollary}

\subsection{Proportionality}

\begin{theorem}\label{theorem:mean_proportionality}
The Mean Rule satisfies proportionality.
\end{theorem}

\begin{proof}
The Mean Rule aggregates votes by summing across all voters, ensuring each project receives funding proportional to its total allocation. Since the final allocation mirrors the proportion of tokens assigned, the rule satisfies proportionality.
\end{proof}

\begin{example}\label{example:median_proportionality}
The Normalized Median Rule does not satisfy proportionality.

Consider three voters with allocations \([1, 0]\), \([1, 0]\), and \([0, 1]\). The median allocation is \([1, 0]\), while proportionality requires \([2/3, 1/3]\), meaning the rule fails proportionality.
\end{example}

\begin{example}\label{example:midpoint_proportionality}
The Midpoint Rule does not satisfy proportionality.

With the same three voters as above, the rule selects \([1, 0]\), minimizing the \(\ell_1\) distance rather than providing \([2/3, 1/3]\), violating proportionality.
\end{example}

\begin{example}\label{example:quadratic_proportionality}
The Quadratic Rule does not satisfy proportionality.

For three voters voting \([9, 1]\), \([5, 5]\), and \([4, 6]\), the rule outputs \([5.6, 4.4]\) instead of the proportional outcome \([6, 4]\), failing proportionality.
\end{example}

\begin{example}\label{example:quorum_proportionality}
The R3 Quorum Median Rule does not satisfy proportionality.

With voters \([1, 0]\), \([1, 0]\), and \([0, 1]\), and a quorum of 2 voters with 1.1 tokens, the rule yields \([1, 0]\) instead of \([2/3, 1/3]\), violating proportionality.
\end{example}

\begin{corollary}\label{example:R4_proportionality}
The R4 Capped Median Rule does not satisfy proportionality as it includes a Quorum \( K3 \) and is median-based.
\end{corollary}

\subsection{Strategyproofness}

\begin{theorem}\label{theorem:mean_strategyproofness}
The Mean Rule is not strategyproof.
\end{theorem}

\begin{example}
Consider two voters with votes \([0.75, 0.25]\) and \([0, 1]\). The Mean Rule outputs \([0.375, 0.625]\), giving voter 1 an $\ell_1$ distance of $0.75$. If voter 1 misreports \([1, 0]\), the new allocation \([0.5, 0.5]\) reduces their $\ell_1$ distance to $0.5$, demonstrating manipulability.
\end{example}

\begin{theorem}\label{theorem:median_strategyproofness}
The Normalized Median Rule is not strategyproof.
\end{theorem}

\begin{example}
Three voters vote \([0.57,\\ 0.24, 0.19]\), \([0.39, 0.48, 0.13]\), and \([0.44, 0.09, 0.48]\). The Normalized Median Rule outputs \([0.506, 0.276, 0.218]\), giving voter 1 an $\ell_1$ distance of $0.1285$. If they misreport \([0.6, 0.2, 0.2]\), the new outcome \([0.524, 0.238, 0.238]\) lowers their $\ell_1$ distance to $0.0962$.
\end{example}

\begin{theorem}\label{theorem:midpoint_strategyproofness}
The Midpoint Rule is not strategyproof.
\end{theorem}

\begin{example}
Three voters vote \([0.9, 0.1]\), \([0.4, 0.6]\), and \([0.2, 0.8]\). The Midpoint Rule outputs \([0.4, 0.6]\), with voter 1's $\ell_1$ distance at $1$. By misreporting \([0.5, 0.5]\), they reduce their distance to $0.8$.
\end{example}

\begin{theorem}\label{theorem:independent_markets_strategyproofness}
The Independent Markets Rule is strategyproof.
\end{theorem}

\begin{proof}
Follows from \cite{freeman2019truthful}[Theorem 4.8] for \( m > 2 \).
\end{proof}

\begin{theorem}\label{theorem:majoritarian_phantom_strategyproofness}
The Majoritarian Phantom Rule is strategyproof.
\end{theorem}

\begin{proof}
Follows from \cite{freeman2019truthful}[Theorem 4.8] for \( m > 2 \).
\end{proof}

\begin{theorem}\label{theorem:quadratic_strategyproofness}
The Quadratic Rule is not SP.
\end{theorem}

\begin{example}
Three voters vote \([0.7, 0.3]\), \([0.4, 0.6]\), and \([0.3, 0.7]\). The Quadratic Rule outputs \([0.483, 0.517]\), with voter 1's $\ell_1$ distance at $0.434$. Misreporting \([0.8, 0.2]\) changes the outcome to \([0.502, 0.498]\), lowering their distance to $0.396$.
\end{example}

\begin{theorem}\label{theorem:quorum_median_strategyproofness}
The R3 Quorum Median Rule is not strategyproof.
\end{theorem}

\begin{example}
Three voters vote \([0.2, 0.8]\), \([0.1, 0.9]\), and \([0.5, 0.5]\) with a quorum of 0.3 tokens and 2 voters per project. The R3 Quorum Median Rule outputs \([0, 1]\), giving voter 1 an $\ell_1$ distance of $0.4$. If they misreport \([0, 1]\), the new allocation \([0.25, 0.75]\) reduces their distance to $0.1$.
\end{example}

\begin{corollary}\label{example:R4_SP}
The R4 Capped Median Rule is not strategyproof since it includes a Quorum \( K3 \) and is median-based.
\end{corollary}

\subsection{Maximal Social Welfare (Minimal L1)}
\label{subsection:maximum_SWF}

An allocation satisfies maximal social welfare (SWF) if no alternative allocation results in a smaller total $\ell_1$ distance between the final allocation and voter ballots.

\begin{theorem}
\label{theorem:mean_SWF}
The Mean Rule does not satisfy maximal SWF.
\end{theorem}

\begin{example}
Three voters vote \((1, 0)\), \((1, 0)\), and \((0.2, 0.8)\). The Mean Rule outputs \((0.733, 0.267)\) with a total $\ell_1$ distance of \(2.133\). An alternative allocation \((1,0)\) reduces the total $\ell_1$ distance to \(1.6\), proving that the Mean Rule is suboptimal.
\end{example}

\begin{remark}
The Mean Rule satisfies maximal SWF if measured by minimizing $\ell_2$ distance \cite{bulteau2021aggregation}.
\end{remark}

\begin{theorem}
\label{theorem:median_SWF}
The Normalized Median Rule does not satisfy maximal SWF.
\end{theorem}

\begin{example}
Three voters vote \((0.1, 0.9)\), \((0.4, 0.2)\), and \((0.6, 0.1)\). The Normalized Median Rule outputs \((0.667, 0.333)\) with a total $\ell_1$ distance of \(1.834\). The alternative allocation \((0.5, 0.5)\) reduces the distance to \(1.7\), showing suboptimality.
\end{example}

\begin{remark}
The Median Rule satisfies maximal SWF if it directly minimizes $\ell_1$ distance without normalization \cite{bulteau2021aggregation}.
\end{remark}

\begin{theorem}
\label{theorem:midpoint_SWF}
The Midpoint Rule does not satisfy maximal SWF.
\end{theorem}

\begin{example}
Four voters vote for four projects: \((0.8,0.1,0.05,0.05)\), \((0.1,0.8,0.05,0.05)\), \((0.05,0.05,0.8,0.1)\), \((0.05,0.05,0.1,0.8)\). The Midpoint Rule selects any voter’s ballot, leading to a total $\ell_1$ distance of \(4.6\). An alternative allocation \((0.25,0.25,0.25,0.25)\) reduces it to \(4.4\), proving suboptimality.
\end{example}

\begin{theorem}
The Independent Markets algorithm does not satisfy maximal SWF, following \cite{freeman2019truthful}[Theorem 6.1].
\end{theorem}

\begin{theorem}
The Majoritarian Phantoms algorithm satisfies maximal SWF, following \cite{freeman2019truthful}[Theorem 6.1].
\end{theorem}

\begin{theorem}\label{theorem:quadratic_SWF}
The Quadratic Voting Rule does not satisfy maximal SWF.
\end{theorem}

\begin{example}
Three voters vote \((0.01, 0.99)\), \((0.01, 0.99)\), and \((0.99, 0.01)\). Quadratic Voting outputs \((0.364,0.636)\) with a total $\ell_1$ distance of \(2.668\). The alternative allocation \((0.2,0.8)\) reduces the distance to \(2.34\), proving suboptimality.
\end{example}

\begin{corollary}
\label{corollary:quorum_median_SWF}
The R3 Quorum Median Rule does not satisfy maximal SWF, as Median-based rules require normalization.
\end{corollary}

\begin{corollary}
\label{example:R4_SWF}
The R4 Capped Median Rule does not satisfy maximal SWF, as it relies on Median-based normalization.
\end{corollary}
\clearpage
\section{Missing Figures}
\begin{figure}[h]
    \centering
    \includegraphics[width=1.1\textwidth]{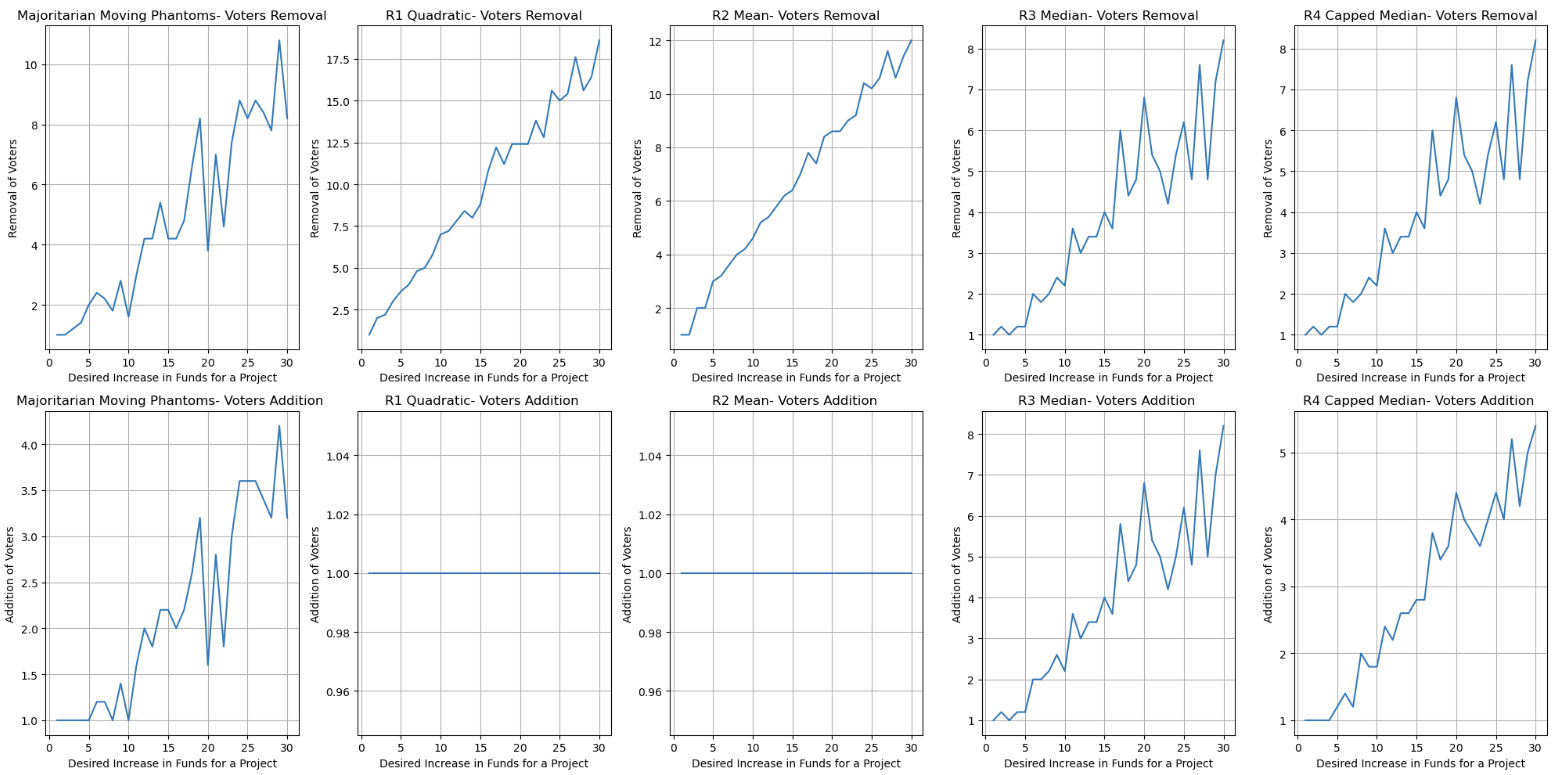}
    \caption{Cost of Adding and deleting voters as a function of the desired funding increase.}
    \label{fig:control}
\end{figure}
\begin{figure}[t]
    \centering
    \includegraphics[width=0.5\textwidth]{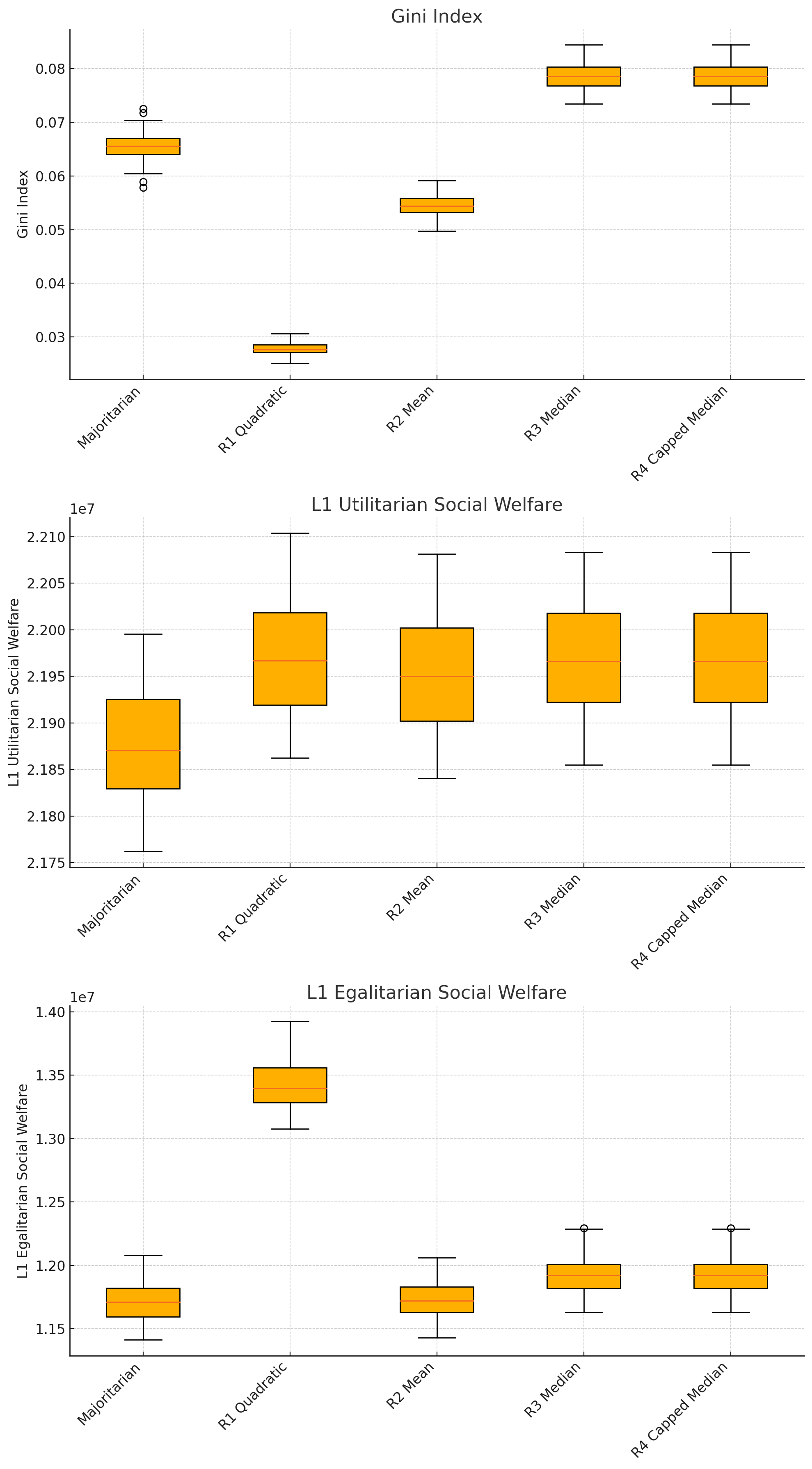}
    \caption{Comparison of Gini Index, L1 Utilitarian Social Welfare, and L1 Egalitarian Social Welfare across voting rules.}
    \label{fig:gini_util_egal}
\end{figure}

\begin{figure}[t]
    \centering
    \includegraphics[width=0.5\textwidth]{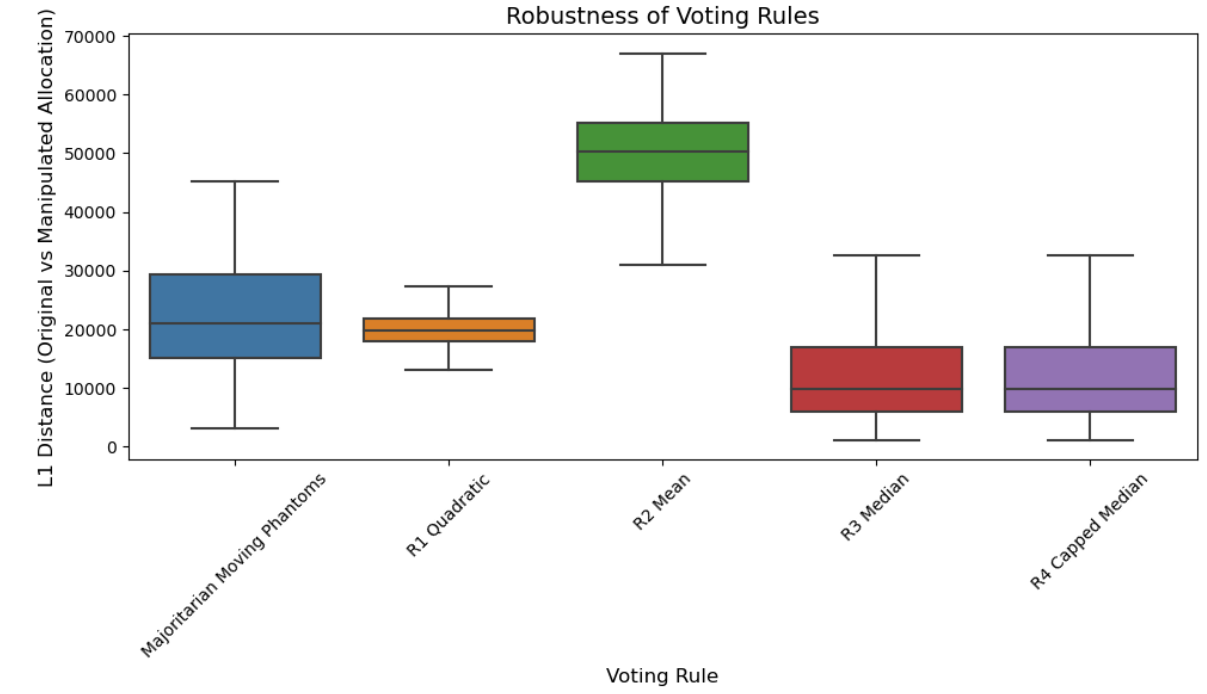}
    \caption{Robustness of voting rules under manipulation, represented by L1 distance between original and manipulated allocations.}
    \label{fig:robustness}
\end{figure}

\begin{figure}[t]
    \centering
    \includegraphics[width=0.5\textwidth]{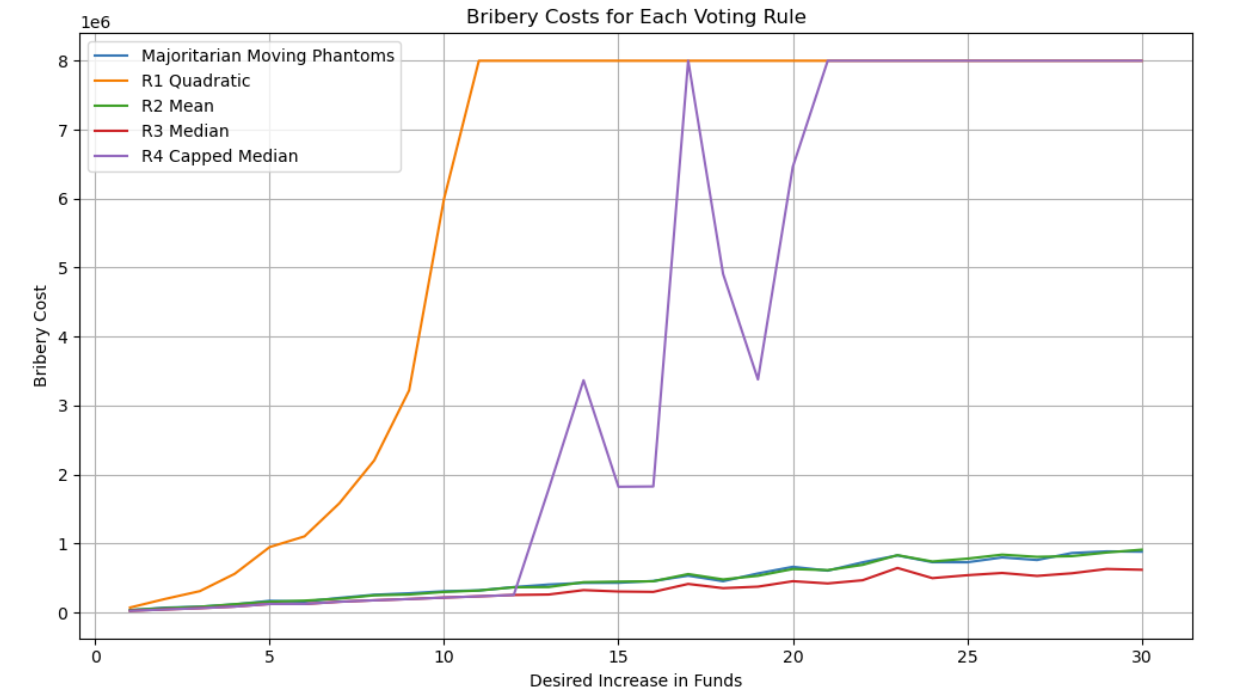}
    \caption{Bribery costs for different voting rules as a function of desired increase in funds for a project.}
    \label{fig:bribery_cost}
\end{figure}

\begin{figure}[t]
    \centering
    \includegraphics[width=0.5\textwidth]{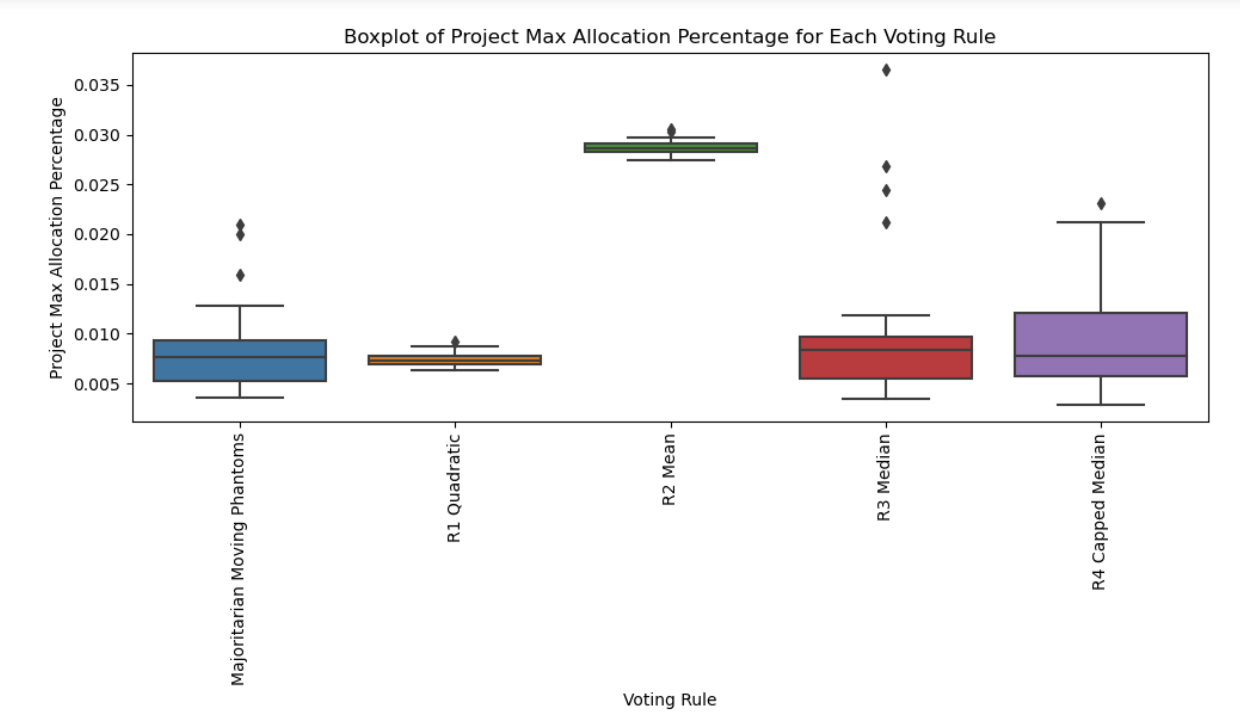}
    \caption{The maximal new project allocation caused by one voter skewing the outcome divided by the total number of tokens to be funded for each voting rule.}
    \label{fig:vev_voters}
\end{figure}

\begin{figure}[t]
    \centering
    \includegraphics[width=0.5\textwidth]{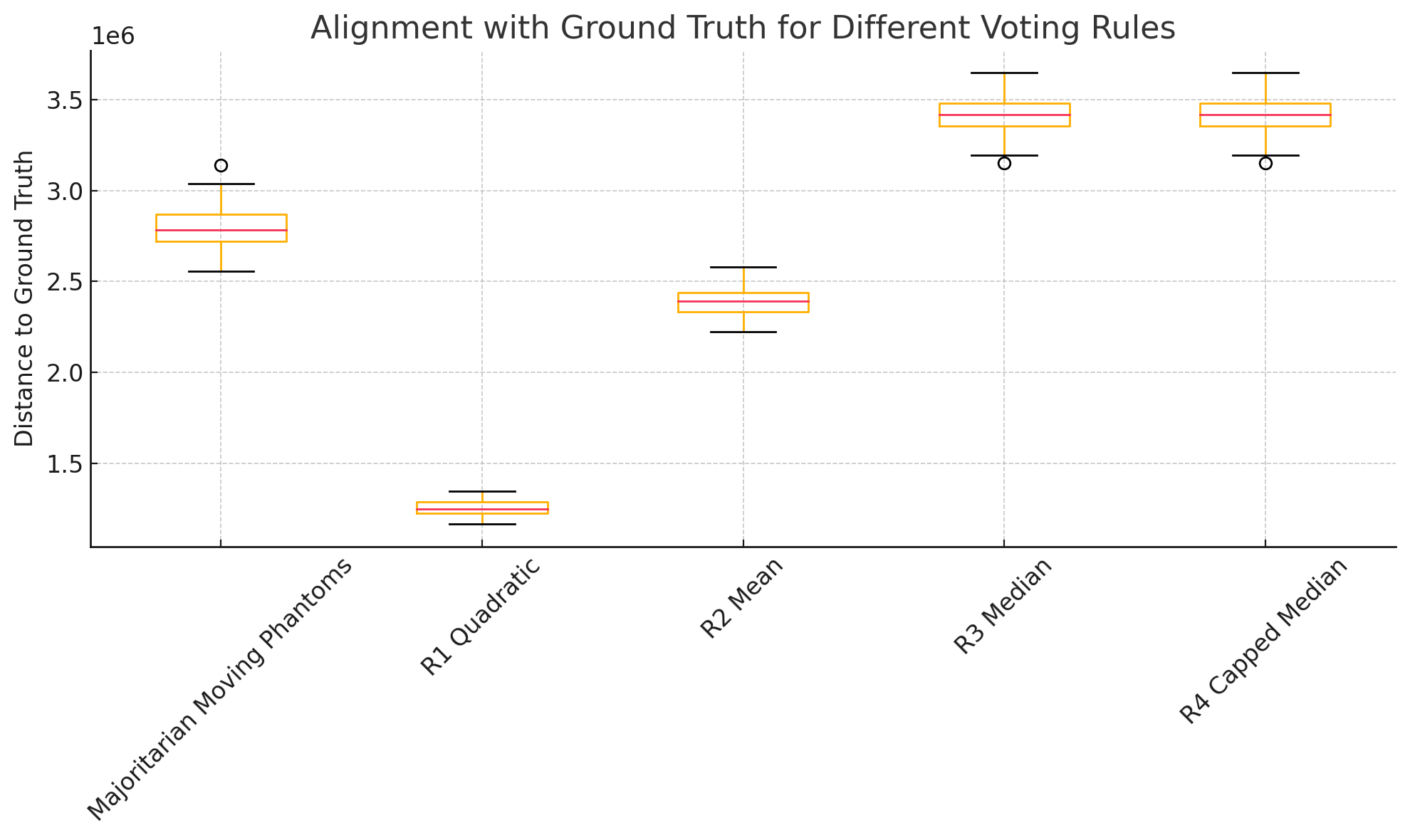}
    \caption{Alignment with ground truth for different voting rules, measured by distance to ground truth.}
    \label{fig:max_vev}
\end{figure}
\begin{figure}[t]
    \centering
    \includegraphics[width=0.5\textwidth]{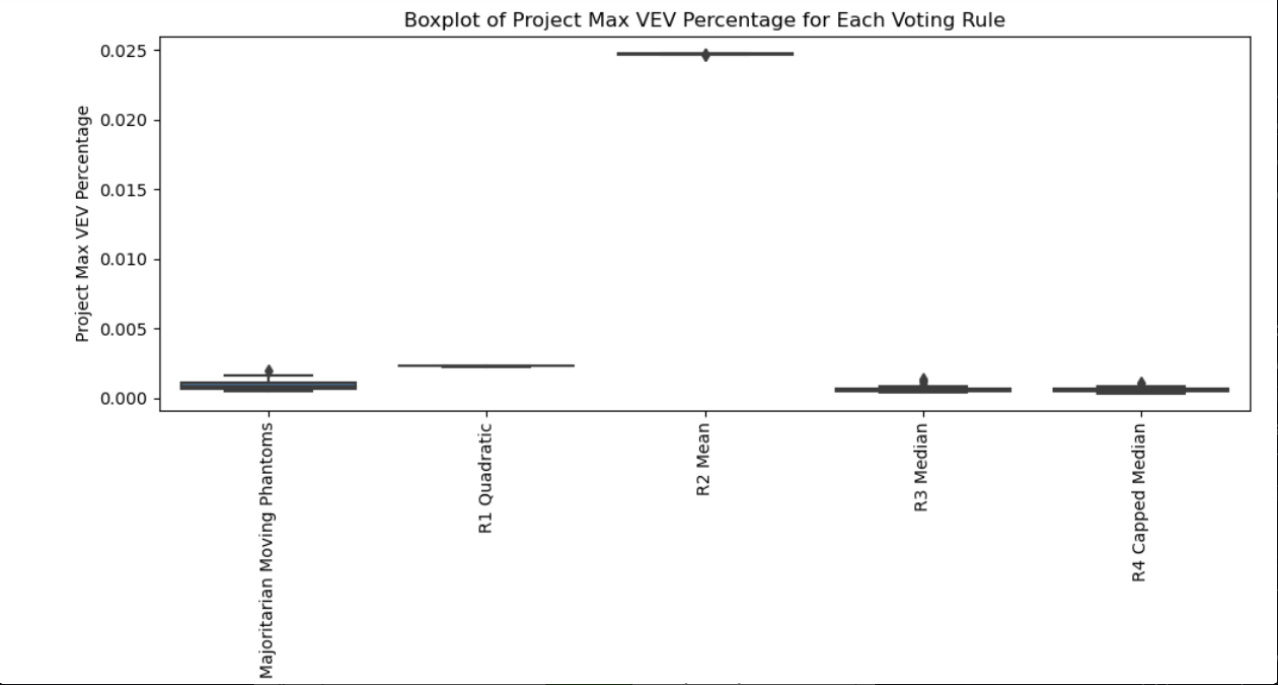}
    \caption{The maximal difference in token allocation caused by one voter skewing the outcome, divided by the total number of tokens to be funded for each voting rule.}
    \label{fig:vev2}
\end{figure}
\begin{figure}[t]
    \centering
    \includegraphics[width=0.4\textwidth]{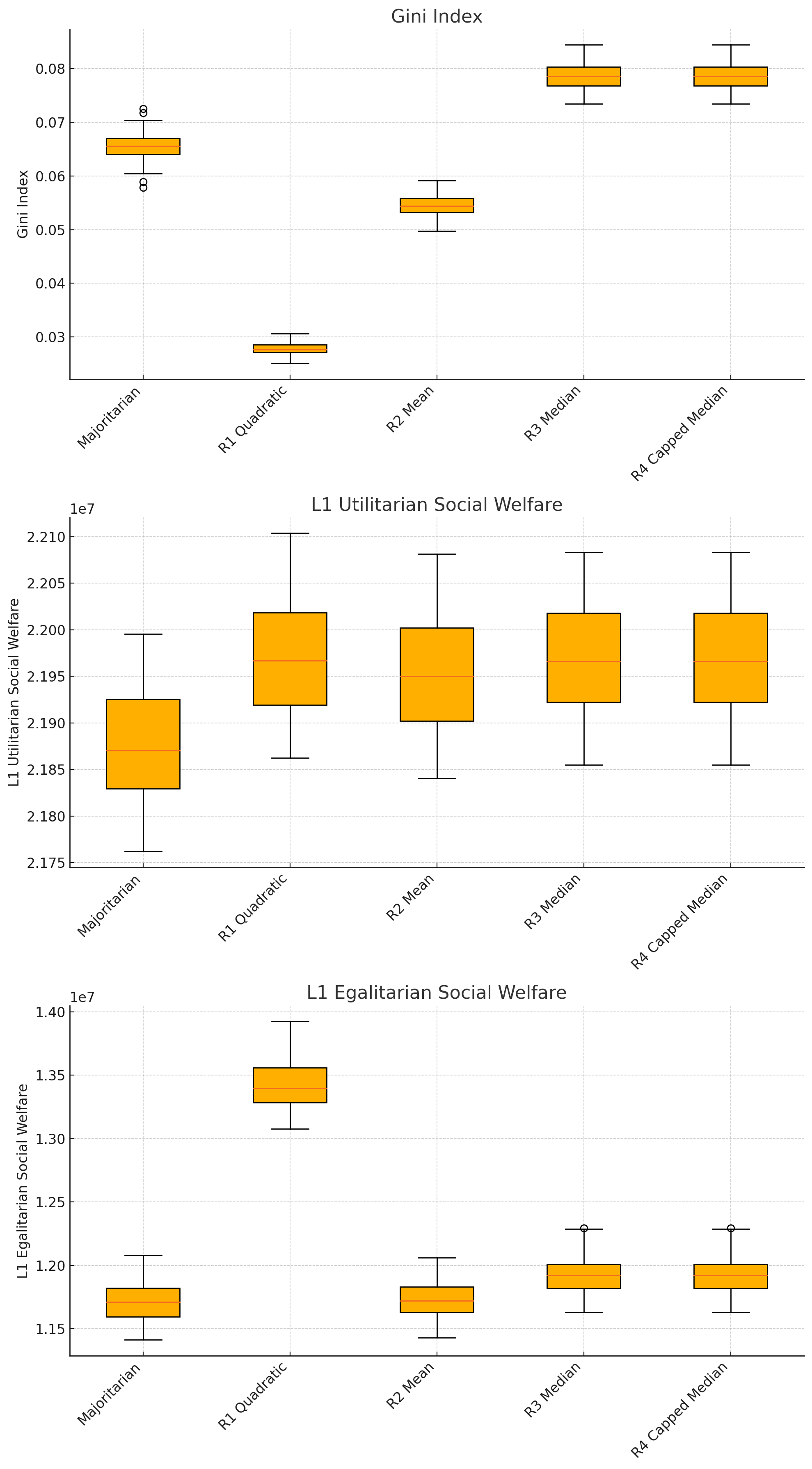}
    \caption{Gini Index, Utilitarian Social Welfare ($\ell_1$ distance), and Egalitarian Social Welfare (maximum $\ell_1$ distance) for each voting rule on round 4 data.}
    \label{fig:Gini_Util_Egal}
\end{figure}
\end{document}